\documentclass{jgaa-art}
\newenvironment{linenomath*}{}{}

\usepackage{amsmath}
\usepackage{amssymb}
\usepackage{caption}
\usepackage{subcaption}

\newtheorem{observation}{Observation}

\newcommand{\canIgnore}[1]{{\tiny (Ignored stuff here)}}
\newcommand{\lex}{\text{lex}}
\newcommand{\skel}[1]{#1^{\text{skel}}}

\newcommand{\calP}{\mathcal{P}}

\newtheorem{definition}{Definition}
\newtheorem{corollary}{Corollary}

\begin{document}
\doi{}
\Issue{0}{0}{0}{0}{0}
\HeadingAuthor{T.Biedl and C.Pennarun}
\HeadingTitle{Non-aligned Drawings} 
\title{Non-aligned Drawings of Planar Graphs}
\Ack{Research supported by NSERC.  A preliminary version appeared at GD'16.}
\author[first]{Therese Biedl}{biedl@uwaterloo.ca}
\author[second]{Claire Pennarun}{claire.pennarun@labri.fr}

\affiliation[first]{David R.~Cheriton School of Computer Science,
University of Waterloo, Waterloo, Canada}
\affiliation[second]{Univ. Bordeaux, 
 CNRS, LaBRI, UMR 5800, F-33400 Talence, France}

\submitted{November 2016}%
\reviewed{}%
\revised{}%
\reviewed{}%
\revised{}%
\accepted{}%
\final{}%
\published{}%
\type{Regular paper}%
\editor{A. Editor}%

\maketitle

\begin{abstract}A {\em non-aligned} drawing of a graph is a drawing
where no two vertices are in the same row or column.
Auber et al.~showed that not all planar graphs
have a non-aligned planar straight-line drawing in the $n\times n$-grid.  
They also showed that such a drawing exists if up to $n-3$ edges may have a bend.

In this paper, we give algorithms for non-aligned planar drawings
that improve on the results by Auber et al.  In particular, we
give such drawings in an $n\times n$-grid with
at most $\frac{2n-5}{3}$ bends, and we study what grid-size can be
achieved if we insist on having straight-line drawings.
\end{abstract}

\Body
\section{Introduction}

At the GD 2015 conference, Auber et al.~\cite{ABDP15} 
introduced the concept of {\em rook-drawings}:  These are
drawings of a graph in an $n\times n$-grid such that no two 
vertices are in the same row or the same column (thus,
if the vertices were rooks on a chessboard, 
then no vertex could beat any other).
They showed that not all planar graphs have a planar straight-line
rook-drawing, and then gave a construction of planar rook-drawings
with at most $n-3$ bends. From now on, all drawings are
required to be planar.

In this paper, we continue the study of rook-drawings.  Note that
if a graph has no straight-line rook-drawing, then we can relax 
the restrictions in two possible ways.  We could either, as Auber
et al.~did, allow to use bends for some of the edges, and try to keep
the number of bends small.  Or we could increase the grid-size and
ask what size of grid can be achieved for straight-line
drawings in which no two vertices share a row or a column; this
type of drawing is known as {\em non-aligned drawing} \cite{AB-GD11}.
A rook-drawing is then a non-aligned drawing on an $n \times n$-grid.

\medskip\noindent{\bf Existing results:}  Apart from the paper by
Auber et al., non-aligned drawings have arisen in a few other contexts.
Alamdari and Biedl showed that every every graph has an inner rectangular
drawing also has a non-aligned drawing \cite{AB-GD11}. These drawings are
so-called rectangle-of-influence drawings and can hence be assumed to be
in an $n\times n$-grid.  In particular, every 4-connected planar
graph with at most $3n-7$ edges therefore has a rook-drawing (see
Section~\ref{sec:4conn} for details).
Non-aligned drawings were also created by Di Giacomo et al. \cite{DDKLM14}
in the context of upward-rightward drawings.  They showed that
every planar graph has a non-aligned drawing in an $O(n^4)\times O(n^4)$-grid.
Finally, there have been studies about drawing graphs with the opposite goal,
namely,
creating as many collinear vertices as possible \cite{LDFMR16}.

\medskip\noindent{\bf Our results:}
In this paper, we show the following (the bounds listed here are upper bounds; see the sections for tighter bounds):

\begin{itemize}
\item Every planar graph has a non-aligned 
	straight-line drawing in an $n^2\times n^2$-grid.
	This is achieved by taking any weak barycentric representation
	(for example, the one by Schnyder \cite{Sch90}),
	scaling it by a big enough factor,
	and then moving vertices slightly so that they have distinct 
	coordinates while maintaining a weak barycentric representation.
\item Every planar graph has a 
	non-aligned straight-line drawing in
	an $n\times \frac{1}{2}n^3$-grid.
	This is achieved by creating drawings with the canonical ordering
	\cite{FPP90} in a standard fashion (similar to \cite{CK97}).  
	However, we pre-compute all the
	$x$-coordinates (and in particular, make them a permutation of
	$\{1,\dots,n\}$), and then argue that with the standard
	construction the slopes do not get too big, and hence the height is
	quadratic.  Modifying the construction a bit, we can also achieve
	that all $y$-coordinates are distinct and that the height is cubic.
\item Every planar graph has a rook-drawing with at most $\frac{2n-5}{3}$ bends.
	This is achieved via creating a so-called rectangle-of-influence
	drawing of a modification of the graph, and arguing that each modification
	can be undone while adding only one bend.
\end{itemize}

Our bounds are even better for 4-connected planar graphs.  In particular,
every 4-connected planar graph has a rook-drawing with at most 1 bend
(and more generally, the number of bends is no more than the number of
so-called filled triangles).  
We also show that any so-called nested-triangle graph has a 
non-aligned straight-line drawing in an
$n\times (\frac{4}{3}n-1)$-grid.

\section{Non-aligned straight-line drawings}

In this section, all drawings are required to be straight-line drawings.

\subsection{Non-aligned drawings on an $n^2 \times n^2$-grid}

We first show how to construct non-aligned drawings in an $n^2\times n^2$-grid
by scaling and perturbing a so-called {\em weak barycentric representation}
(reviewed below).

In the following, a vertex $v$ is assigned to a triplet of
non-negative integer coordinates 
$(p_0(v),p_1(v),p_2(v))$.  For two vertices $u,v$ and $i=0,1,2$, we say that
$(p_i(u),p_{i+1}(u)) <_{\lex} (p_i(v), p_{i+1}(v))$ if
either $p_i(u) < p_i(v)$, or $p_i(u) = p_i(v)$ and $p_{i+1}(u) < p_{i+1}(v)$. Note that in this section, addition on the subscripts is done modulo 3.

\begin{definition}[Weak barycentric representation \cite{Sch90}]
A \emph{weak barycentric representation} of a graph $G$ is an injective function $\calP$ that maps each $v\in V(G)$ to a point $(p_0(v),p_1(v),p_2(v)) \in \mathbb{N}_0^3$ such that
\begin{itemize}
\item $p_0(v) + p_1(v) + p_2(v) = c$ for every vertex $v$, where
$c$ is a constant independent of the vertex,
\item for each edge $(u,v)$ and each vertex $w \neq \{u,v\}$, there is some $k \in \{0,1,2\}$ such that 
$(p_k(u), p_{k+1}(u)) <_{\lex} (p_k(z), p_{k+1}(z))$ and
$(p_k(v), p_{k+1}(v)) <_{\lex} (p_k(z), p_{k+1}(z))$.
\end{itemize}
\end{definition}

\begin{theorem}[\cite{Sch90}]
\label{thm:schnyder}
Every planar graph with $n$ vertices has a weak barycentric representation with $c=n-1$.
Furthermore, $0\leq p_i(v) \leq n-2$ for all vertices $v\in V$ and all $i\in \{0,1,2\}$.
\end{theorem}

Observe that weak barycentric representations are preserved
under scaling, i.e., if we have a weak barycentric representation $\calP$ 
(say with constant $c$),
then we can scale all assigned coordinates by the same factor $N$
and obtain another weak barycentric representation (with
constant $c\cdot N$).  We need to do slightly more, namely scale and ``twist'',
as detailed in the following lemma.

\begin{lemma}
\label{lem:transform_barycentric}
Let $G$ be a graph with a weak barycentric representation $\calP$ with $\calP=\left(\left(p_0(v),p_1(v),p_2(v)\right)_{v\in V}\right)$.
Let $N\geq 1+ \max_{v\in V} \{\max_{i=0,1,2} p_i(v)\}$ be a positive integer.  Define $\calP'$
to be the assignment $p'_i(v) := N\cdot p_i(v)+p_{i+1}(v)$ for $i=0,1,2$.
Then $\calP'$ is also a weak barycentric representation.
\end{lemma}

\begin{proof}
We need to check the following properties:
\begin{itemize}
\item[(a)] For some constant $c$ we have $p'_1(v) + p'_2(v) + p'_3(v) = c$ for all vertices $v$.
\item[(b)] $\calP'$ is injective.
\item[(c)] For each edge $(u,v)$ and each vertex $z \neq \{u,v\}$, there is some $k \in \{0,1,2\}$ such that 
$(p'_k(u), p'_{k+1}(u)) <_{\lex} (p'_k(z), p'_{k+1}(z))$ and
$(p'_k(v), p'_{k+1}(v)) <_{\lex} (p'_k(z), p'_{k+1}(z))$.
\end{itemize}

(a) Let $c_P$ be the constant of $\calP$.  Then
for each vertex $v$, $p'_1(v) + p'_2(v) + p'_3(v) = N\left(p_1(v) + p_2(v) + p_3(v)\right) + p_1(v) + p_2(v) + p_3(v) = N \cdot c_P + c_P$, which is a constant.

\medskip
(b) Let $\{u,v\}$ be two vertices of $G$, $u \neq v$. 
Since $\calP$ is injective, we know that there exists $i \in \{0,1,2\}$ such that $p_i(u) \neq p_i(v)$. Without loss of generality, $p_i(u) > p_i(v)$. Since all coordinates $p_i$ are integers, $p_i(u) \geq p_i(v) +1$. Thus $N \cdot p_i(u) \geq N \cdot p_i(v) + N > N\cdot p_i(v) + p_{i+1}(v)-p_{i+1}(u)$ by 
$p_{i+1}(v)<N$ and
$p_{i+1}(u)\geq 0$.
Thus $p'_{i}(u) > p'_{i}(v)$ and $\calP'$ is injective.

\medskip
(c) Let $(u,v)$ be an edge of $G$ and $z \neq \{u,v\}$ a vertex of $G$. Since $\calP$ is a weak barycentric representation, there is some $k \in \{0,1,2\}$ such that 
$(p_k(u), p_{k+1}(u)) <_{\lex} (p_k(z), p_{k+1}(z))$ and $(p_k(v), p_{k+1}(v)) <_{\lex} (p_k(z), p_{k+1}(z))$.

We only show the claim for $u$, and have two cases:
\begin{itemize}
\item $p_k(u) < p_k(z)$: As in part (b), then $p'_k(u)<p'_k(z)$.
\item $p_k(u) = p_k(z)$: Then $p_{k+1}(u) < p_{k+1}(z)$ 
and
$p'_k(u)=N p_k(u)+p_{k+1}(u) = N p_k(z)+p_{k+1}(u) < N p_k(z)+p_{k+1}(z) = p'_{k}(z)$.
\end{itemize}
So either way $p'_k(u)<p'_k(z)$ and hence 
$(p'_k(u), p'_{k+1}(u)) <_{\lex} (p'_k(z), p'_{k+1}(z))$.
\end{proof}

Applying this to Schnyder's weak barycentric representation, we now have:

\begin{theorem}
\label{thm:non-aligned}
Every planar graph has a non-aligned straight-line planar drawing in an $(n(n-2))\times (n(n-2))$-grid.
\end{theorem}

\begin{proof}
Let $\calP=\left(\left(p_0(v),p_1(v),p_2(v)\right)_{v\in V}\right)$ be the weak barycentric representation of 
Theorem~\ref{thm:schnyder}; we know that $0\leq p_i(v)\leq n-2$ for all $v$ and all $i$.
Now apply Lemma~\ref{lem:transform_barycentric} with $N=n-1$
to obtain the weak barycentric representation $\calP'$ with
$p'_i(v)=(n-1)p_i(v)+p_{i+1}(v)$.  Observe that $p'_i(v)\leq (n-1)(n-2)+(n-2)=n(n-2)$.
Also, $p'_i(v)\geq 1$ since not both $p_i(v)$ and $p_{i+1}(v)$ can be 0.   (More precisely,
$p_i(v)=0=p_{i+1}(v)$ would imply $p_{i+2}(v)=n-1$, contradicting $p_{i+2}(v)\leq n-2$.)

As shown by Schnyder \cite{Sch90}, mapping each vertex $v$ to point
$(p_0'(v),p_1'(v))$ gives a planar straight-line drawing of $G$.  By the
above, this drawing has the desired grid-size.  It remains
to show that it is non-aligned, i.e.,
for any two vertices $u,v$ and any $i \in \{0, 1\}$, we have $p'_i(u) \neq p'_i(v)$.
Assume after possible renaming that $p_i(u)\leq p_i(v)$.  We have two cases:
\begin{itemize}
\item If $p_i(u) < p_i(v)$, then $p_i(u) \leq p_i(v) -1$ since $\calP$ assigns integers. Thus 
$N \cdot p_i(u) \leq N \cdot p_i(v) - N < N\cdot p_i(v) - p_{i+1}(u) + p_{i+1}(v)$
 since $p_{i+1}(u)<N$ and $p_{i+1}(v)\geq 0$.  Therefore $p'_i(u) < p'_i(v)$.
\item If $p_i(u) = p_i(v)$, then $p_{i+1}(u) \neq p_{i+1}(v)$ (else the three coordinates of $u$ and $v$ would be the same, which is impossible since $\calP$ is an injective function). Then $p'_i(u) = N \cdot p_i(u) + p_{i+1}(u) \neq N \cdot p_i(v) + p_{i+1}(v)= p'_i(v)$.
\end{itemize}
\end{proof}

\subsection{Non-aligned drawings on an $n \times f(n)$-grid}

We now show how to build non-aligned drawings for which the width
is the minimum-possible $n$, and the height is
$\approx \frac{1}{2}n^3$.  We use the well-known
canonical ordering for {\em triangulated
plane graphs}, i.e., graphs for which the planar embedding is
fixed and all faces (including the outer-face) are triangles.
We hence assume throughout that $G$ is triangulated; we can
achieve this by adding edges and delete them
in the obtained drawing.

The {\em canonical ordering} \cite{FPP90} of such a graph
is a vertex order $v_1,\dots,v_n$ such that $\{v_1,v_2,v_n\}$
is the outer-face, and for any $3\leq k\leq n$, the graph $G_k$
induced by $v_1,\dots,v_k$ is 2-connected.  This implies that
$v_k$ has at least 2 {\em predecessors} (i.e., neighbours in $G_{k-1}$), 
and its predecessors form an interval on the outer-face of $G_{k-1}$.  
We assume (after possible renaming) that $v_1$ is the neighbour of $v_2$ found in clockwise order
on the outer-face, and enumerate the outer-face of graph $G_{k-1}$
in clockwise order
as $c_1,\dots,c_L$ with $c_1=v_1$ and $c_L=v_2$.  Then the predecessors
of $v_k$ consist of $c_\ell,\dots,c_r$ for some $1\leq \ell < r \leq L$;
we call $c_\ell$ and $c_r$ the {\em leftmost} and {\em rightmost}
predecessors of $v_k$ (see also Figure~\ref{fig:canonical}).
In this section, $x(v)$ and $y(v)$ denote the $x$- and $y$-coordinates
of a vertex $v$, respectively.

\subsubsection{Distinct $x$-coordinates}

We first give a construction that achieves distinct $x$-coordinates
in $\{1,\dots,n\}$ (but $y$-coordinates may coincide).
Let $v_1,\dots,v_n$ be a canonical ordering.
The goal is to build a straight-line drawing
of the graph $G_k$ induced by $v_1,\dots,v_k$ using induction on $k$.  
The key idea is to define {\em all} $x$-coordinates beforehand.
Orient the edges of $G$ as follows.  Direct $(v_1,v_2)$ as
$v_1\rightarrow v_2$.  For $k\geq 3$, if $c_r$ is the 
rightmost predecessor of $v_k$, then direct all edges
from predecessors of $v_k$ towards $v_k$, with the exception of $(v_{k},c_r)$,
which is directed $v_{k}\rightarrow c_r$.

By induction on $k$, one easily shows that the orientation of $G_k$ is 
acyclic, with unique source $v_1$ and unique sink $v_2$, and the outer-face 
directed $c_1\rightarrow \dots \rightarrow c_L$.%
\footnote{For readers familiar with how a Schnyder wood $T_1,T_2,T_3$  
can be obtained from a canonical ordering: The orientation is the same
as $T_1^{-1}\cup T_2\cup T_3$, and hence acyclic \cite{Sch90}.}
Find a topological order $x: V\rightarrow \{1,\dots,n\}$ of 
the vertices, i.e., if $u\rightarrow v$ then $x(u)<x(v)$.  We use this
topological order as our $x$-coordinates, and hence have $x(v_1)=1$ and $x(v_2)=n$.   
(We thus have two distinct vertex-orderings:  one defined by
the canonical ordering, which is used to compute $y$-coordinates, and one
defined by the topological ordering derived from the canonical ordering, which 
directly gives the $x$-coordinates.)

Now construct a drawing of $G_k$ that respects these $x$-coordinates by
induction on $k$ (see also Figure~\ref{fig:addition2}).
Start with $v_1$ at $(1,2)$, $v_3$ at $(x(v_3),2)$ and $v_2$ at $(n,1)$.

For $k\geq 3$, let $c_\ell$ and $c_r$ be the leftmost and rightmost 
predecessors of $v_{k+1}$.  Notice that $x(c_\ell)<\dots<x(c_r)$ due
to our orientation, which in particular implies that for any $\ell\leq j\leq r$,
the upward ray from $c_j$ intersects no other vertex or edge.
Let $y^*$ be the smallest integer value such
that any $c_j$, for $\ell\leq j\leq r$, can ``see'' the point 
$p=(x(v_{k+1}),y^*)$
in the sense that that the line segment from $c_j$ to $p$
intersects no other vertices or edges.  Such a $y^*$ exists since 
the upward ray from $c_j$ is empty: by tilting this ray slightly,
$c_j$ can also see all sufficiently high points on
the vertical line $\{x=x(v_{k+1})\}$. 
Placing $v_{k+1}$ at
$(x(v_{k+1}),y^*)$ hence gives a planar drawing of $G_{k+1}$,
and we continue until we get a drawing of $G_n=G$.

\begin{figure}[ht]
\hspace*{\fill}
\begin{subfigure}[b]{0.48\linewidth}
\includegraphics[page=1,width=\linewidth,trim=20 0 0 20,clip]{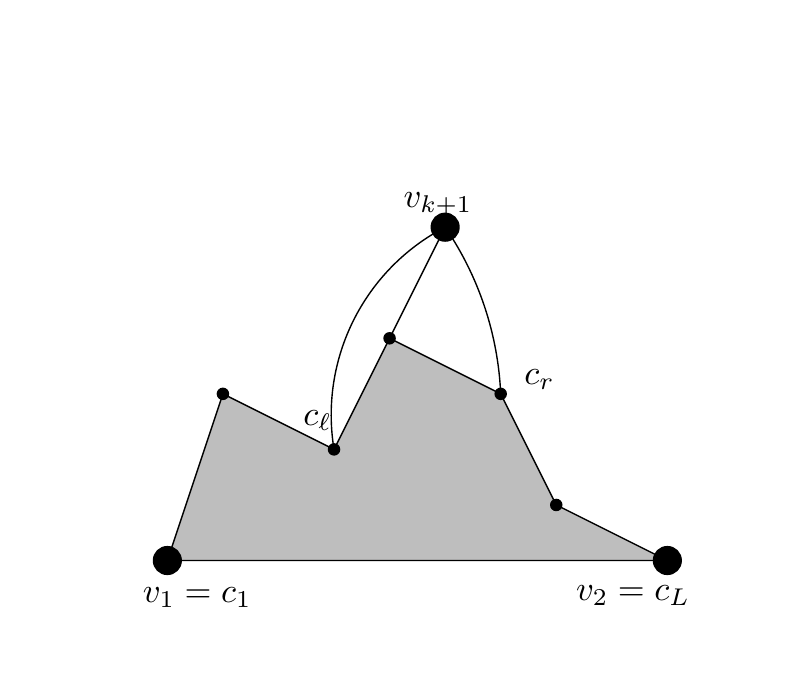}
\caption{Illustration of a canonical order.}
\label{fig:canonical}
\end{subfigure}
\hspace*{\fill}
\begin{subfigure}[b]{0.48\linewidth}
\includegraphics[page=2,width=\linewidth,trim=20 0 0 20,clip]{canonical_order.pdf}
\caption{Finding a $y$-coordinate for $v_{k+1}$.}
\label{fig:addition2}
\end{subfigure}
\hspace*{\fill}
\caption{Drawing algorithm to find distinct $x$-coordinates in $\{1,\dots,n\}$.}
\end{figure}

To analyze the height of this construction, we bound the slopes.

\begin{lemma}
\label{lem:slopes1}
Define $s(k):=k-3$ for $k\geq 3$.     
All edges on the outer-face of the constructed drawing of $G_k$ have slope at most $s(k)$ for $k
\geq 3$.
\end{lemma}

\begin{proof}
Clearly this holds for $k=3$ and $s(3)=0$.  Now assume it holds for some $k\geq 3$, and
let $c_\ell,\dots,c_r$ be the predecessors of $v_{k+1}$.   Fix one predecessor
$c_j$ for $\ell\leq j<r$, and 
consider the ray $\rho_j$ of slope $s(k)$ starting from $c_j$.  Since 
all edges in $G_k$ have slope at most $s(k)$, vertex $c_j$ can see all
points that are above $\rho_j$ and to the right of $c_j$.    In particular,
consider therefore the point where ray $\rho_\ell$ intersects
the vertical line $\{x=x(v_{k+1})\}$, 
and set $y'$ to be the smallest integer $y$-coordinate that is strictly
above this intersection.    By construction, point $p=(x(v_{k+1}),y')$ is above $\rho_j$
and to the right of $c_j$ for $j=\ell,\dots,r-1$, and hence can see all of them.

We claim that point $p$ can also see $c_r$.
This holds because edge $(c_{r-1},c_r)$ has slope at most $s(k)$,
and (due to the chosen edge directions)
$x(c_{r-1})<x(v_{k+1})<x(c_r)$.  Therefore point $p$ is above $(c_{r-1},c_r)$,
and can be connected to both of them without intersection.  Also note that 
line segment $(p,c_r)$ therefore has a smaller slope than $(c_{r-1},c_r)$,
and in particular slope less than $s(k)$.

So point $(x(v_{k+1}),y')$ can see all vertices $c_\ell,\dots,c_r$,
and the value of $y^*=y(v_{k+1})$ is no bigger than $y'$.  We already
argued that edge $(v_{k+1},c_r)$ has slope less than $s(k)\leq s(k+1)$, so we
only must argue the slope of the unique other new outer-face edge 
$(c_\ell,v_{k+1})$.
Since $y'$ is the smallest integer $y$-coordinate above the point
where $\rho_\ell$
intersects the line $\{x=x(v_{k+1})\}$, we have
\begin{linenomath*}
\begin{equation}
\label{equ:y1}
y^*\leq y' \leq 
y(c_l) + \left(x(v_{k+1})-x(c_l)\right)\cdot s(k) + 1.
\end{equation}
\end{linenomath*}
By $x(v_{k+1})-x(c_l)\geq 1$
the slope of $(c_l,v_{k+1})$ is at most
\begin{linenomath*}
\begin{equation*}
\frac{y^*-y(c_l)}{x(v_{k+1})-x(c_l)}
\leq
s(k)+\frac{1}{x(v_{k+1})-x(c_l)} \leq s(k)+1 = s({k+1})
\end{equation*}
\end{linenomath*}
as desired.
\end{proof}

Vertex $v_n$ has $x$-coordinate at most $n-1$, and
the edge from $v_1$ to $v_n$ has slope at most $s(n)=n-3$.
This shows that the $y$-coordinate of $v_n$ 
is at most $2+(n-2) \cdot (n-3)$.  Since triangle $\{v_1,v_2,v_n\}$
bounds the drawing, this gives:

\begin{theorem}
Every planar graph has a planar straight-line drawing in an $n\times (2+(n-2)(n-3))$-grid
such that all vertices have distinct $x$-coordinates.
\end{theorem}

While this theorem per se is not useful for non-aligned drawings, we find
it interesting from a didactic point of view:  It proves that polynomial
coordinates can be achieved for straight-line drawings of planar graphs, 
and requires for this only the canonical ordering, but neither the 
properties of Schnyder trees \cite{Sch90} nor the details of how to ``shift''
that is 
needed for other methods using the canonical ordering (e.g.~\cite{CK97,FPP90}).
We believe that our bound on the height
is much too big, and that the true height is $o(n^2)$ and possibly $O(n)$.

\subsubsection{Non-aligned drawings}

We now modify the above construction slightly to achieve
distinct $y$-coordinates.
Define the exact same $x$-coordinates and place $v_1$ and $v_2$ as before.  
To place vertex $v_{k+1}$, let $y^*$ be the smallest $y$-coordinate such
that point $(x(v_{k+1}),y^*)$ can see all predecessors of $v_{k+1}$, and
such that none of $v_1,\dots,v_k$ is in row $\{y=y^*\}$. Clearly this
gives a non-aligned drawing.
It remains to bound how much this increases the height.

\begin{figure}[ht]
\hspace*{\fill}
\includegraphics[page=3,width=0.48\linewidth]{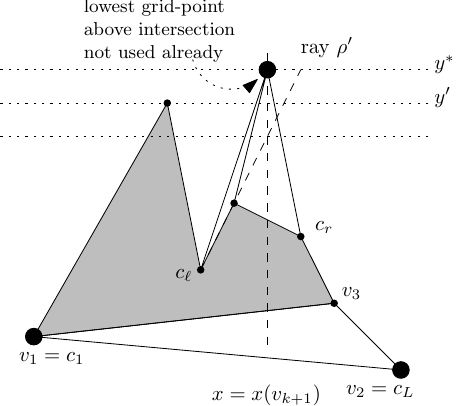}
\hspace*{\fill}
\caption{Finding a $y$-coordinate for $v_{k+1}$ that has not been
used by earlier vertices.}
\end{figure}

\begin{lemma}
\label{lem:slopes2}
Define $s'(k):=\sum_{i=1}^{k-2} i = \frac{1}{2}(k-1)(k-2)$ for $k\geq 3$.     
All edges on the outer-face of the constructed 
non-aligned drawing
 of $G_k$ have slope at most $s'(k)$ for $k
\geq 3$.
\end{lemma}
\begin{proof}
The claim clearly holds for $k=3$, since $v_3$ is placed with $y$-coordinate 3
and therefore $(v_1,v_3)$ has slope at most $1=s(3)$.
Now let $k\geq 3$ and consider the time when adding $v_{k+1}$ with predecessors 
$c_\ell,\dots,c_r$, and define $\rho'$ to be the ray of slope $s'(k)$
emanating from $c_\ell$.  Let $y'$ be the smallest integer
coordinate above the intersection of $\rho'$ with the vertical line $\{x=x(v_{k+1})\}$.
As in Lemma~\ref{lem:slopes1}, one argues that $p'=(x(v_{k+1}),y')$ can see all
of $c_\ell,\dots,c_r$.  

We may or may not be able to use point $p'$ for $v_{k+1}$, depending on
whether some other vertices were in the row $\{y=y'\}$.
Observe that $y'\geq 3$, because $y(c_\ell)\geq 2$ and $s'(k)\geq 1$.  
Therefore neither $v_1$ nor $v_2$ had $y$-coordinate $y'$, which leaves
at most $k-2$ vertices that could be in row $y'$ or higher.
In particular therefore
\begin{linenomath*}
\begin{equation}
\label{equ:y2}
y^* \leq y'+(k-2) \leq y(c_\ell) + \left(x(v_{k+1})-x(c_\ell)\right)\cdot s'(k) + 1+(k-2)
\end{equation}
\end{linenomath*}
Reformulating as before shows that the slope of $(c_\ell,v_{k+1})$ is at most
\begin{linenomath*}
\begin{equation*}
\frac{y^*-y(c_l)}{x(v_{k+1})-x(c_l)}
\leq
s'(k)+\frac{k-1}{x(v_{k+1})-x(c_l)} \leq s'(k)+k-1 = s'({k+1}).
\end{equation*}
\end{linenomath*}
\end{proof}

Edge $(v_1,v_n)$ has slope at most $\frac{1}{2}(n-1)(n-2)$.
Since $x(v_n)-x(v_1)\leq n-2$ and $y(v_1)=2$, therefore
the height is at most $2+\frac{1}{2}(n-1)(n-2)^2$. 

\begin{theorem}
Every planar graph has a 
non-aligned straight-line drawing in
an $n\times \left(2+\frac{1}{2}(n-1)(n-2)^2\right)$-grid.
\end{theorem}

Comparing this to Theorem~\ref{thm:non-aligned}, we see that
the aspect ratio is much worse, but 
the area is smaller.
We suspect that the method results in a smaller height than the 
proved upper bound: Equation~\eqref{equ:y2} is generally not
tight, and so a smaller slope-bound
(implying a smaller height) is likely to hold.  

\subsection{The special case of nested triangles}

We now turn to non-aligned drawings of a special graph class.
Define a {\em nested-triangle graph} $G$ as follows.  $G$ has $3k$
vertices for some $k\geq 1$, say $\{u_i,v_i,w_i\}$ for $i=1,\dots,k$.
Vertices $\{u_i,v_i,w_i\}$ form a triangle (for $i=1,\dots,k$).
We also have paths $u_1,u_2,\dots,u_k$ as well as $v_1,v_2,\dots,v_k$
and $w_1,w_2,\dots,w_k$.  With this the graph is 3-connected; we assume
that its outer-face is $\{u_1,v_1,w_1\}$.  All interior faces that
are not triangles may or may not have a diagonal in them, and there are
no restrictions on which diagonal (if any).
Nested-triangle graphs are of interest in graph drawing because they
are the natural lower-bound graphs for the area of straight-line drawings 
\cite{DLT84}.

\begin{theorem}
Any nested-triangle graph with $n=3k$ vertices has a
non-aligned
straight-line drawing in an $n\times (\frac{4}{3}n-1)$-grid.
\end{theorem}
\begin{proof}
The 4-cycle $\{w_k,v_k,v_{k-1},w_{k-1}\}$ may or may not have a diagonal
in it; after possible exchange of $w_1,\dots,w_k$ and $v_1,\dots,v_k$
we assume that there is no edge between $v_{k-1}$ and $w_k$.
For $i=1,\dots,k$, place $u_i$ at $(i,i)$,
vertex $v_i$ at $(3k+1-i, k+i)$, and $w_i$ at 
$(k+i,4k+1-2i)$ (see Figure~\ref{fig:stacked}).
The $x$- and $y$-coordinates are all distinct. The $x$-coordinates range from 1 to $n$,
and the maximal $y$-coordinate is $4k-1 = \frac{4}{3}n-1$.
It is easy to check that all interior faces are drawn strictly
convex, with the exception of $\{v_k,v_{k-1},w_{k-1},w_k\}$
which has a $180^\circ$ angle at $v_k$, but our choice of naming
ensured that there is no edge $(v_{k-1},w_k)$.
Thus any diagonal inside these 4-cycles
can be drawn without overlap.
Since $G$ is planar, two edges joining vertices of different triangles cannot cross.
Thus $G$ is drawn without crossing in an $n \times \left( \frac{4}{3}n-1 \right)$-grid.
\end{proof}

\begin{figure}
\begin{center}
\includegraphics[scale=0.7]{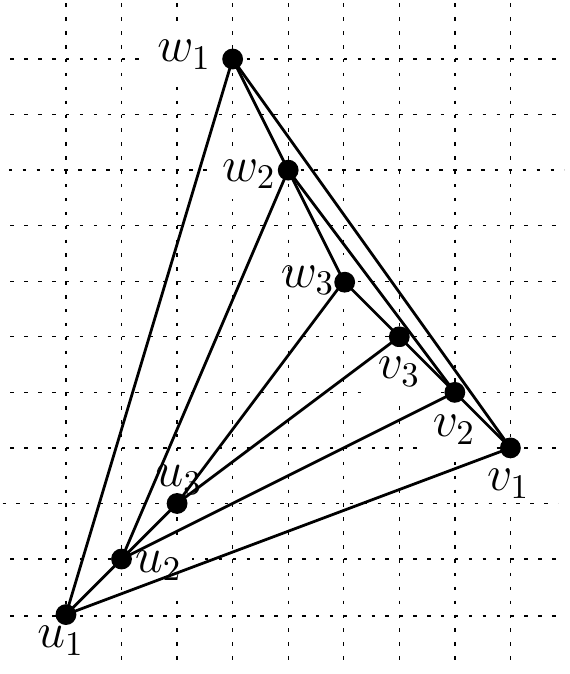}
\caption{A non-aligned straight-line drawing of a nested-triangle 
graph with $k=3$ on an $9 \times 11$-grid.
}
\label{fig:stacked}
\end{center}
\end{figure}

In particular, notice that the octahedron is a nested-triangle graph
(for $k=2$) and this construction gives a  
non-aligned straight-line drawing in a 
$6\times 7$-grid.  This is clearly optimal since it has no straight-line
rook-drawing \cite{ABDP15}.

We conjecture that this construction gives the minimum-possible height
for nested-triangle graphs among all non-aligned straight-line drawings.

\section{Rook-drawings with bends} 

We now construct rook-drawings with bends; as before
we do this only for triangulated graphs.
The main idea is to find rook-drawings with
only 1 bend for 4-connected triangulated graphs, then convert
any graph into a 4-connected triangulated graph by subdividing few
edges and re-triangulating, and finally argue that the drawing for
it, modified suitably, gives a rook-drawing with few bends.

We need a few definitions first.  Fix a triangulated graph $G$.
A {\em separating triangle} is a triangle
that has vertices both strictly inside and strictly outside the triangle.
$G$ is {\em 4-connected} (i.e., cannot be made disconnected
by removing 3 vertices) if and only if it has no separating triangle.
A {\em filled triangle} \cite{BBM-GD99} of $G$ is a triangle that has
vertices strictly inside.  A triangulated graph 
has at least one filled triangle (namely, the outer-face)
and every separating triangle is also a filled triangle.
We use $f_G$ to denote the number of filled triangles of
the graph $G$.

A {\em rectangle-of-influence (RI) drawing} is a 
straight-line drawing such that for any edge $(u,v)$, the
minimum axis-aligned rectangle containing $u$ and $v$
is {\em empty}, i.e. contains no other vertex of the drawing 
in its relative interior~\footnotemark~.
The following is known:
\footnotetext{
In the literature there are four kinds of RI-drawings, 
depending on whether points on the boundary of the rectangle
are allowed or not (open vs.~closed RI-drawings), 
and whether an edge $(u,v)$ must exist if $R(u,v)$ 
is empty (strong vs.~weak
RI-drawings).  The definition here corresponds to open weak RI-drawings.
}

\begin{figure}
\begin{center}
\begin{subfigure}[b]{0.3\linewidth}
\includegraphics[scale=0.43,page=7]{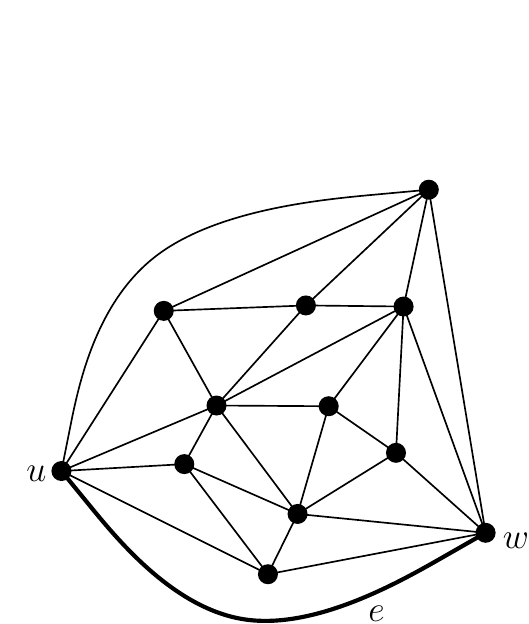} 
\caption{A triangulated graph $G$; $e$ is one of the edges of its outer-face.} 
\end{subfigure}
\hfill
\begin{subfigure}[b]{0.3\linewidth}
\includegraphics[scale=0.43,page=15,trim=0 20 0 0,clip]{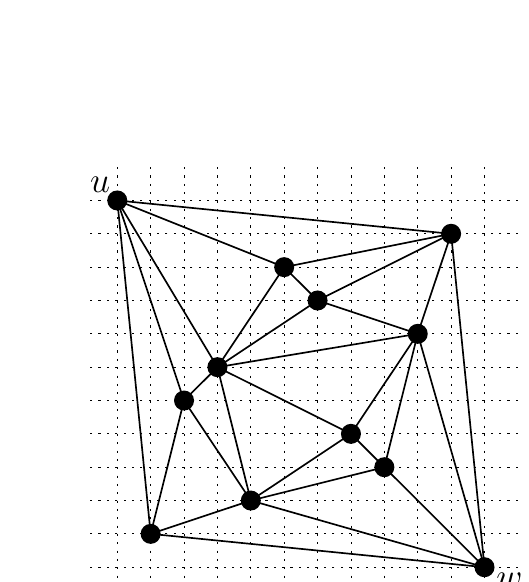} ~~
\caption{An RI-drawing of $G-e$ satisfying the conditions of Theorem~\ref{thm:BBM}.}
\label{fig:RI_example}
\end{subfigure}
\hfill
\begin{subfigure}[b]{0.3\linewidth}
\includegraphics[scale=0.5,page=4]{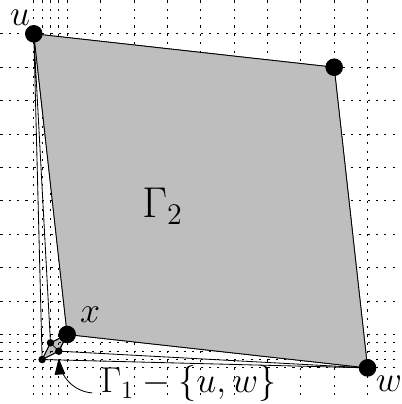} 
\caption{Combining two RI-drawings if all separating triangles contain $(u,w)$.  }
\label{fig:merge_RI}
\end{subfigure}
\caption{RI-drawings.}
\end{center}
\end{figure}


\begin{theorem}[\cite{BBM-GD99}]
\label{thm:BBM}
Let $G$ be a triangulated 4-connected graph and let $e$ be an edge on the
outer-face.  Then $G-e$ has a planar RI-drawing.  

Moreover, the drawing is non-aligned and on an $n\times n$-grid, 
the ends of $e$ are at $(1,n)$ and $(n,1)$,
and the other two vertices on the outer-face are at $(2,2)$ and $(n-1,n-1)$.
\end{theorem}

Figure~\ref{fig:RI_example} illustrates such a drawing of a graph. 
The latter part of the claim is not specifically stated in \cite{BBM-GD99},
but can easily be inferred from the construction (see also a 
simpler exposition in \cite{BD-31order}).

RI-drawings are useful because they can
be deformed (within limits) without introducing crossings.
We say that two drawings $\Gamma$ and $\Gamma'$ of a graph have
{\em the same relative coordinates} if for any two vertices $v$ and $w$,
we have $x_\Gamma(v)<x_\Gamma(w)$ if and only if
$x_{\Gamma'}(v)<x_{\Gamma'}(w)$,
and $y_\Gamma(v)<y_\Gamma(w)$ if and only if
$y_{\Gamma'}(v)<y_{\Gamma'}(w)$, where $x_\Gamma(v)$ denotes the
$x$-coordinate of $v$ in $\Gamma$, etc.
The following result appears to be folklore; we sketch a proof for completeness.

\begin{observation}
\label{obs:RI_deform}
Let $\Gamma$ be an RI-drawing.  If $\Gamma'$ is a straight-line drawing with
the same relative coordinates as $\Gamma$, then $\Gamma'$ is an
RI-drawing, and it is planar if and only if $\Gamma$ is.
\end{observation}
\begin{proof}
The claim on the RI-drawing was shown by Liotta et al.~\cite{LLMW98}.
It remains to argue planarity.  Assume that edge $(u,v)$ crosses edge $(w,z)$
in an RI-drawing.  Since all rectangles-of-influence are empty, this 
happens if and only if 
(up to renaming) we have $x(w)\leq x(u)\leq x(v)\leq x(z)$
and $y(u)\leq y(w)\leq y(z)\leq y(v)$.   This only depends on the
relative orders of $u,v,w,z$, and hence a transformation maintaining
relative coordinates also maintains planarity.
\end{proof}

We need a slight strengthening of Theorem~\ref{thm:BBM}.

\begin{lemma}
\label{lem:extendBBM}
Let $G$ be a triangulated graph, let $e\in E$ be an edge on the outer-face,
and assume all separating triangles of $G$ contain $e$.  
Then $G-e$ has a planar RI-drawing. 
Moreover, the drawing is non-aligned and on an $n\times n$-grid, 
the ends of $e$ are at $(1,n)$ and $(n,1)$,
and the other two vertices on the outer-face are at $(2,2)$ and $(n-1,n-1)$.
\end{lemma}
\begin{proof}
We proceed by induction on the number of separating triangles of $G$.
In the base case, $G$ is 4-connected and the claim holds by
Theorem~\ref{thm:BBM}.  For the inductive step, assume that $T=\{u,x,w\}$ is
a separating triangle.  By assumption it contains $e$, say $e=(u,w)$.
Let $G_1$ be the graph consisting of $T$
and all vertices inside $T$, and let $G_2$ be the graph obtained
from $G$ by removing all vertices inside $T$.  Apply induction
to both graphs.  In drawing $\Gamma_2$ of $G_2-e$, vertex $x$ is on
the outer-face and hence (after possible reflection) placed at $(2,2)$.  Now insert a (scaled-down)
copy of the drawing $\Gamma_1$ of $G_1$, minus vertices $u$ and $w$, in the
square $(1,2]\times (1,2]$ (see Figure~\ref{fig:merge_RI}).
Since $x$ was (after possible reflection) in the top-right corner of 
$\Gamma_1-\{u,w\}$, the two copies of
$x$ can be identified.  One easily verifies that this gives an RI-drawing,
because within each drawing the relative coordinates are unchanged, 
and the two drawings have disjoint $x$-range and $y$-range
except at $u$ and $w$.  Finally, re-assign coordinates to the vertices 
while keeping relative coordinates intact so that we have an $n\times n$-grid;
by Observation~\ref{obs:RI_deform} this gives a planar RI-drawing.
\end{proof}

\subsection{4-connected planar graphs} \label{sec:4conn}

Combining Theorem~\ref{thm:BBM} with Observation~\ref{obs:RI_deform}, 
we immediately obtain:

\begin{theorem}
Let $G$ be a triangulated 4-connected planar graph.
Then $G$ has a planar rook-drawing with at most one bend.
\end{theorem}
\begin{proof}
Fix an arbitrary edge $e$ on the outer-face, and apply Theorem~\ref{thm:BBM}
to obtain an RI-rook-drawing $\Gamma$ of $G-e$.
It remains to add in edge $e=(u,v)$.
One end $u$ of $e$ is in the top-left corner, and the
leftmost column contains no other vertex.  The other end $v$ is
in the bottom-right corner, and the bottommost row contains no
other vertex.  We can hence route $(u,v)$ by going vertically
from $u$ and horizontally from $v$, with the bend in the bottom-left corner.
\end{proof}

\begin{corollary}
Let $G$ be a 4-connected planar graph.  Then $G$ has a rook-drawing
with at most one bend, and with no bend if $G$ is not triangulated.
\end{corollary}
\begin{proof}
If $G$ is triangulated then the result was shown above, so assume
$G$ has at least one face of degree 4 or more.
Since $G$ is 4-connected,
one can add edges
to $G$ such that the result $G'$ is triangulated and 4-connected
\cite{BKK97}.  Pick a face incident to an added edge $e$ as outer-face of $G'$, 
and apply Theorem~\ref{thm:BBM} to obtain an
RI-drawing of $G'-e$.  
Deleting all edges in $G'-G$ gives the result.
\end{proof}

Since we have only one bend, and the ends of the edge $(u,v)$ that contain
it are the top-left and bottom-right corner, we can remove the bend by
stretching. 

\begin{theorem}
Every 4-connected planar graph has 
a non-aligned planar drawing in an $n\times (n^2-3n+4)$-grid and in a $(2n-2) \times (2n-2)$-grid.
\end{theorem}
\begin{proof}
Let $\Gamma$ be the RI-drawing of $G-(u,v)$ with $u$ at $(1,n)$ and $v$ at
$(n,1)$.  Relocate $u$ to point $(1,n^2-3n+4)$.
  The resulting drawing is still
a planar RI-drawing by Observation~\ref{obs:RI_deform}.  Now 
$y(u)-y(v)=(n-2)(n-1)+1$, hence the line segment
from $u$ to $v$ has slope less than $-(n-2)$, and is therefore above 
point $(n-1,n-1)$ (and with that, also above all other vertices of
the drawing).
So we can add this edge without violating planarity,
and obtain a 
non-aligned straight-line drawing of $G$
(see Figure~\ref{fig:4connectedSemiRook}).

For the other result, start with the same drawing $\Gamma$.  Relocate
$u$ to $(1,2n-2)$ and $v$ to $(2n-2,1)$.  The line segment from $u$ to
$v$ has slope $-1$ and crosses $\Gamma$ only between points $(n-1,n)$
and $(n-1,n)$, where no points of $\Gamma$ are located.
So we obtain a non-aligned 
planar straight-line drawing (see Figure~\ref{fig:4connected2nx2n}).
\end{proof}

\begin{figure}[t]
\hspace*{\fill}
\begin{subfigure}[b]{0.35\linewidth}
\includegraphics[scale=0.4]{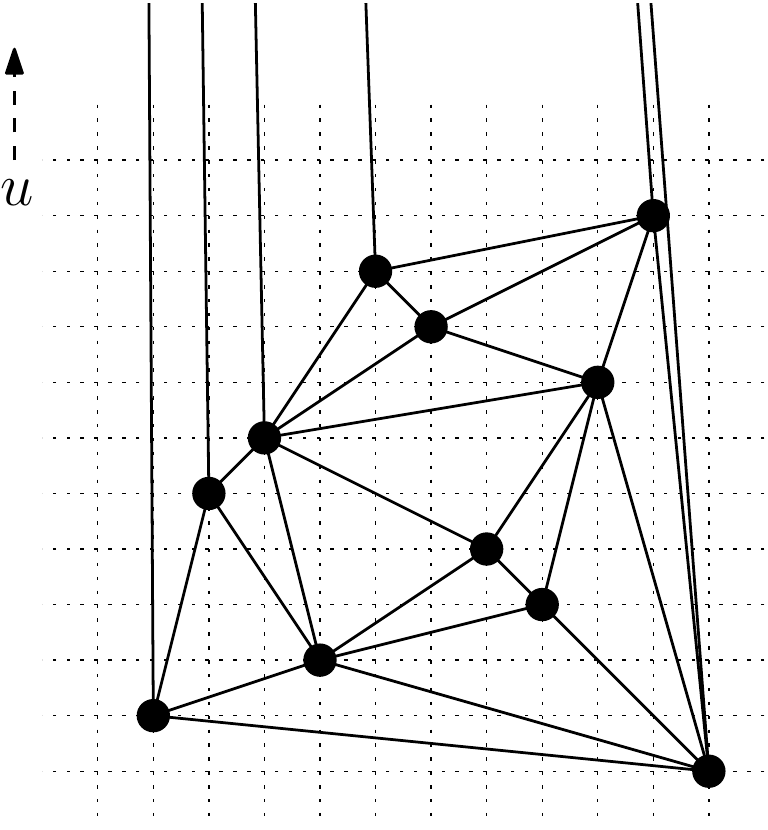}
\caption{A non-aligned drawing of width $n$.}
\label{fig:4connectedSemiRook}
\end{subfigure}
\hspace*{\fill}
\begin{subfigure}[b]{0.35\linewidth}
\includegraphics[scale=0.4]{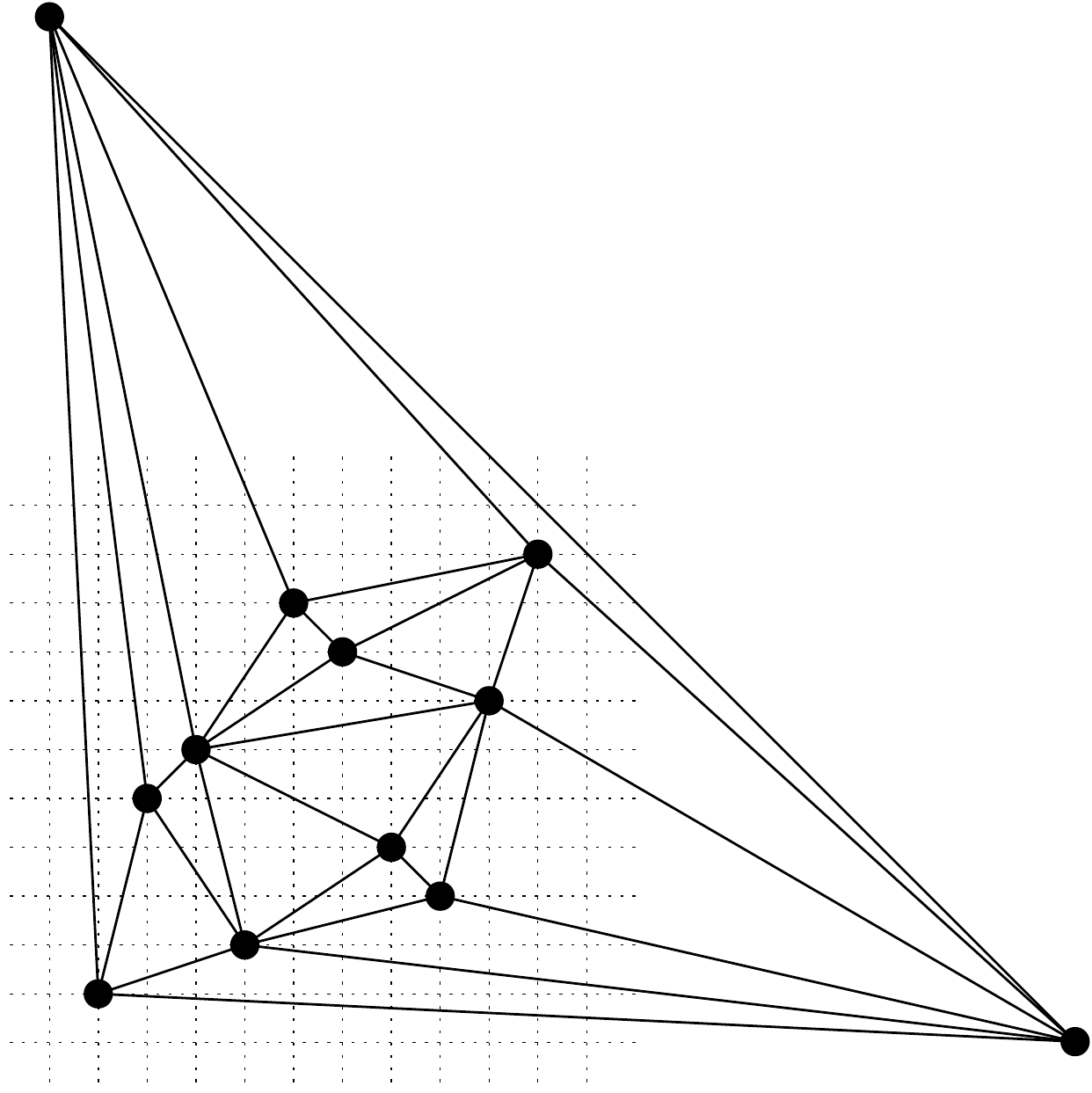}
\caption{A non-aligned drawing on a $(2n-2)\times(2n-2)$-grid.}
\label{fig:4connected2nx2n}
\end{subfigure}
\hspace*{\fill}
\caption{
Transforming Figure~\ref{fig:RI_example} into straight-line drawings.}
\end{figure}

\subsection{Constructing rook-drawings with few bends} 

We now explain the construction of a (poly-line) rook-drawing
for a triangulated graph $G$ with at least 5 vertices.
We proceed as follows: 

\begin{enumerate}
\item Find a small independent-filled-hitting set $E_f$.

Here, an {\em independent-filled-hitting set} 
is a set of edges $E'$
such that (i) every filled triangle has
at least one edge in $E'$ (we say that $E'$ {\em hits} all 
filled triangles), and (ii) every face of $G$ has at most one
edge in $E'$ (we say that $E'$ is {\em independent}).
We can show the following bound on $|E'|$:

\begin{lemma}
\label{lem:independent-filled-hitting}
Any triangulated graph
$G$ of order $n$ has an independent-filled-hitting set of size at most 
\begin{itemize}
\item $f_G$ (where $f_G$ is the number of
filled triangles of $G$), and it can be found in $O(n)$ time,
\item 
$\frac{2n-5}{3}$, and it can be found in $O((n\log n)^{1.5}\sqrt{\alpha(n,n)})$ time
or approximated arbitrarily close in $O(n)$ time.  Here $\alpha$ is the
slow-growing inverse Ackermann function.
\end{itemize}
\end{lemma}

The proof of this lemma requires detours into matchings and 4-coloring;
to keep the flow of the algorithm-explanation we therefore defer it
to the appendix (Section~\ref{sec:independentFilled}).

\item Since the outer-face is a filled triangle, there exists one edge $e_o\in E_f$ 
	that belongs to the outer-face.  Define $E_s:=E_f-\{e_o\}$ and notice
	that $E_s$ contains no outer-face edges since $E_f$ is independent.
\item As done in some previous papers \cite{KW02,CHK+15}, 
	remove separating triangles
	by subdividing all edges $e\in E_s$, and re-triangulate 
	by adding edges from the subdivision vertex
	(see Figure~\ref{fig:subdivide}).
	Let $V_x$ be the new set of vertices,
	and let $G_1$ be the new graph.  Observe that $G_1$ may still have
	separating triangles, but all those separating triangles contain $e_o$
	since $E_f$ hits all filled triangles.

\begin{figure}
\begin{center}
\hspace*{\fill}
\includegraphics[scale=0.75,page=1]{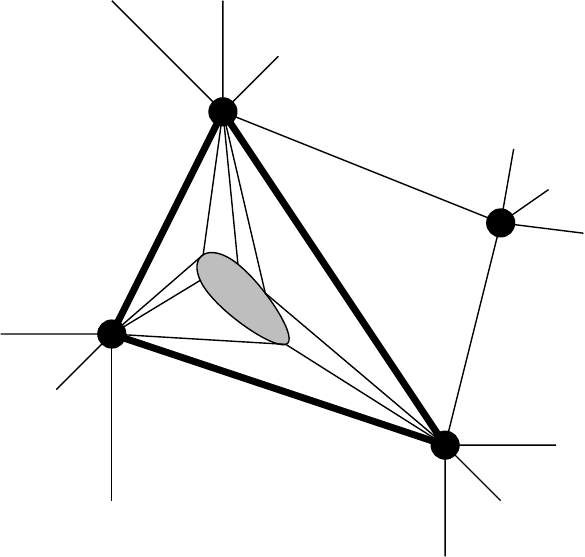}
\hspace*{\fill}
\includegraphics[scale=0.75,page=2]{subdivide.pdf}
\hspace*{\fill}
\caption{
A separating triangle removed by subdividing one of its edges and re-triangulating.
}
\label{fig:subdivide}
\end{center}
\end{figure}

\item	By Lemma~\ref{lem:extendBBM},  $G_1-e_o$
	has a non-aligned RI-drawing $\Gamma$ 
	where the ends of $e_o$ are at the top-left and
	bottom-right corner.
\item Transform $\Gamma$ into drawing $\Gamma'$ so that the relative
	orders stay intact, the original vertices (i.e.,
	vertices of $G$) are
	on an $n\times n$-grid and the subdivision vertices (i.e., vertices
	in $V_x$) are in-between.  

This can be done by enumerating the vertices in $x$-order, and assigning
new $x$-coordinates in this order, increasing to the next integer for
each original vertex and increasing by $\frac{1}{|V_x|+1}$ for each subdivision vertex.
Similarly update the $y$-coordinates (see Figure~\ref{fig:RI_transform}).
	Drawing $\Gamma'$  is still a non-aligned RI-drawing,
	and the ends of $e_o$ are on the top-left and bottom-right 
	corner.
	
\item Let $e$ be an edge in $E_s$ with subdivision vertex $x_e$.
	Since $e$ is an interior edge of $G$, $x_e$ is an interior vertex 
	of $G_1$.  Now move $x_e$ to some integer grid-point nearby.
	This is possible due to the following.

	\begin{lemma}
	\label{lem:move_RI}
	Let $\Gamma$ be a planar RI-drawing. 
	Let $x$ be an interior vertex of degree 4 with neighbours 
	$u_1,u_2,u_3,u_4$ that form a 4-cycle.
	Assume that none of $x,u_1,u_2,u_3,u_4$ share a grid-line.
	Then we can move $x$ to a point on grid-lines of its neighbours
	and obtain a planar RI-drawing.
	\end{lemma}

	The proof of this lemma is not hard, but requires careful checking
	of all positions of neighbours of $x$; we defer it to the 
	appendix (Section~\ref{sec:move_RI}).

	Note that the neighbours of $x_e$ are not in $V_x$, since $E_s$ is independent.
	So we can apply this operation independently to all subdivision-vertices.
\item
	Now replace each subdivision-vertex
	$x_e$ by a bend, 
	connected to the ends of $e$ along the corresponding edges from $x_e$
	(see Figure~\ref{fig:RI_shift}).
	(Sometimes, as is the case in the example, we could also simply
	delete the bend and draw edge $e$ straight-line.)
	None of the shifting changed positions for vertices of $G$, so we
	now have a rook-drawing of $G-e_o$ with bends.  The above shifting
	of vertices does not affect outer-face vertices, so the ends of $e_o$ 
	are still in the top-left and bottom-right corner.  As the final
	step draw $e_o$
	by drawing vertically from one end and horizontally from the other;
	these segments are not occupied by the rook-drawing.

\end{enumerate}

\begin{figure}
\begin{center}
\hspace*{\fill}
\begin{subfigure}[b]{0.45\linewidth}
\includegraphics[width=\linewidth,page=1,trim=0 0 190 0,clip]{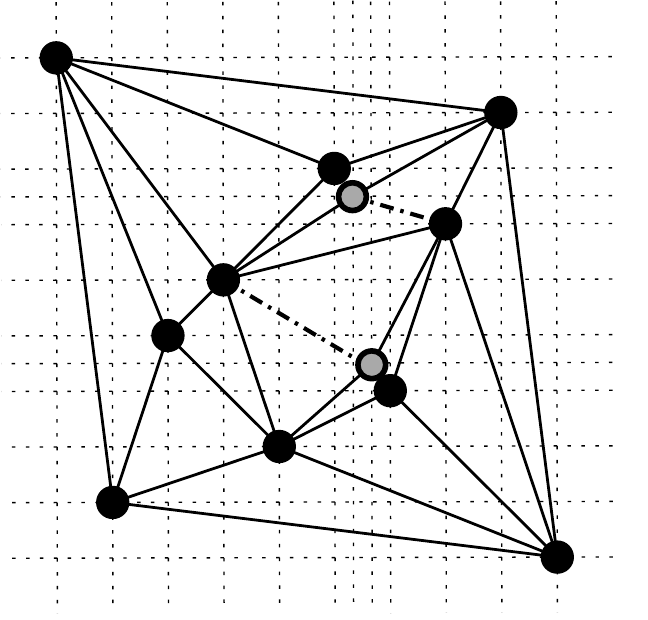}
\caption{Reorder such that subdivision vertices (grey) are not at grid points.}
\label{fig:RI_transform}
\end{subfigure}
\hspace*{\fill}
\begin{subfigure}[b]{0.45\linewidth}
\includegraphics[width=\linewidth,page=1,trim=190 0 0 0,clip]{placing_vertices.pdf}
\caption{Subdivision vertices shifted to integer grid-points, and adding $e_0$.}
\label{fig:RI_shift}
\end{subfigure}
\hspace*{\fill}
\caption{Creating non-aligned drawings with few bends.}
\end{center}
\end{figure}

We added one bend for each edge in $E_f$.
By Lemma~\ref{lem:independent-filled-hitting}, we can an $E_f$
with $|E_f|\leq f_G$ and $|E_f|\leq \frac{2n-5}{3}$ (neither bound
is necessarily smaller than the other), and hence have:

\begin{theorem}
Any planar graph $G$ of order $n$ has a planar rook-drawing with at most $b$ bends, with
$b\leq \min\{ \frac{2n-5}{3}, f_G\}$.
\end{theorem}

Remark that the algorithm proposed to compute such a drawing 
has polynomial run-time, but is not linear-time
due to the time complexity of finding a small
independent-filled-hitting set. All other steps (including computing an
 RI-drawing) can be done in linear time, so if we are content with a
bound of $b\leq \min\{ f_G, (1-\varepsilon) \frac{2n-5}{3}\}$ for
arbitrarily small $\varepsilon$ then the drawing can be found in linear time.


\section{Conclusion}

In this paper, we continued the work on planar rook-drawings initiated by Auber
et al.~\cite{ABDP15}.  We constructed planar rook-drawings with at most $\frac{2n-5}{3}$
bends; the number of bends can also be bounded by the number of filled triangles.  We also
considered drawings that allow more rows and columns while keeping vertices
on distinct rows and columns; we proved that such non-aligned planar straight-line
drawings always exist and have area $O(n^4)$.
As for open problems, the most interesting question is lower bounds.  No planar
graph is known that needs more than one bend in a planar rook-drawing, and no planar
graph is known that needs more than $2n+1$ grid-lines in a planar non-aligned drawing.
The ``obvious'' approach of taking multiple copies of the octahedron fails 
because the property of having a rook-drawing is not closed under taking subgraphs:
if vertices are added, then they could ``use up'' extraneous grid-lines 
in the drawing of a subgraph.  We conjecture that the $n\times (\frac{4}{3}n-1)$-grid
achieved for nested-triangle graphs is optimal for planar straight-line
non-aligned drawings with width $n$.

\newcommand{\student}[1]{#1}
\newcommand{\postdoc}[1]{#1}

\bibliographystyle{plain}
\bibliography{papers.bib,full.bib,gd.bib}

\iftrue
\newpage
\begin{appendix}
\section{Independent-filled-hitting sets}
\label{sec:independentFilled}

Recall that we want to find a set $E_f$ that {\em hits} all filled
triangles (i.e., contains at least one edge of each filled triangles) 
and is {\em independent} (i.e., no face contains two edges of $E_f$).

Our first result shows how to find a matching of size at most 
$f_G$ in linear time.  The existence of such
a matching could easily be proved using the 4-color theorem (see below
for more details), but with a different approach we can find it 
in linear time.

\begin{lemma}
Any triangulated
planar graph $G$ has an independent-filled-hitting set $E_f$ of size at most 
$f_G$.  It can be found in linear time.
\end{lemma}
\begin{proof}
We prove a slightly stronger statement, namely, that we can find such a set
$E_f$ and additionally (i) we can prescribe which edge $e_o$ on the outer-face 
is in $E_f$, and (ii) every separating triangle has {\em exactly} one edge 
in $E_f$.

We proceed by induction on the number of filled triangles.  If $G$ has only 
one (namely, the outer-face), then use the prescribed edge $e_o$; this satisfies
all claims.
Now assume $G$ has multiple filled triangles, and
let $T_1,\dots,T_k$ (for $k\geq 1$) be the maximal separating triangles in the
sense that no other separating triangle contains $T_i$ inside.
Define $G_i$ (for $i=1,\dots,k$) to be the graph consisting
of $T_i$ and all vertices inside $T_i$.  Since we chose maximal separating 
triangles, graphs $G_1,\dots,G_k$ are disjoint. 
Let the {\em skeleton} $\skel{G}$ of $G$
be the graph obtained from $G$ by removing the interior of $T_1,\dots,T_k$.

Let $(\skel{G})^*$ be the dual graph of $\skel{G}$, i.e., it has a vertex 
for every face of $\skel{G}$ and a dual edge $e^*$ for any edge $e$
that connects the two faces that $e$ is incident to.  Since $\skel{G}$
is triangulated, its dual graph is 3-regular and 3-connected,
and therefore has a perfect matching $M$ by Petersen's theorem (see
e.g.~\cite{BBD+01}).  
If the dual edge $e_o^*$ of $e_o$ is not in $M$, then find an alternating cycle that contains $e_o^*$ and swap matching/non-matching edges so that $e_o^*$ is in $M$.  

For every maximal separating triangle $T_i$ of $G$, exactly one edge $e_i$ of $T_i$ has its dual edge in matching $M$,
since $T_i$ forms a face in $\skel{G}$.  Find an independent-filled-hitting set $E_f(G_i)$ 
of $G_i$ that contains $e_i$ recursively.  Combine all these independent-filled-hitting sets
into one set and add $e_o$ to it (if not already in it); the result is set $E_f$.  
Every filled triangle of $G$ is either the outer-face or a filled triangle
of one of the subgraphs $G_i$, so this hits all filled triangles.  Also, every 
filled triangle contains exactly one edge of $E_f$, so $|E_f|$ is as desired.
Finally if $f$ is a face of $G$, then it is either an inner face of one of the subgraphs $G_i$ or a face of $\skel{G}$. Either way at most one
edge of $f$ is in $E_M$; therefore $E_f$ is independent.

The time complexity is dominated by splitting the graph into its 4-connected
components at all separating triangles, which can be done in
linear time \cite{Kant-IJCGA97},
and by finding the perfect matching in the dual graph, which
can also be done in linear time \cite{BBD+01}.
\end{proof}

We now give the second result, which gives a different (and sometimes better)
bound at the price of being slower to compute.
Recall that the 4-color-theorem for planar graphs states that we
can assign colors $\{1,2,3,4\}$ to vertices of $G$ such that no edge has the
same color at both endpoints \cite{AH77a}.  Define $M_1$ to be all those edges
where the ends are colored $\{1,2\}$ or $\{3,4\}$, set $M_2$ to be all those
edges where the ends are colored $\{1,3\}$ or $\{2,4\}$ and set $M_3$ to be
all those edges where the ends are colored $\{1,4\}$ or $\{2,3\}$.  Since every
face of $G$ is a triangle and colored with 3 different colors, the following is
easy to verify:

\begin{observation}
For $i=1,2,3$, edge set $M_i$ contains exactly
one edge of each triangle.
\end{observation}

So each $M_i$ is an independent-filled-hitting set.  Define $E_i$ to be the
set of edges obtained by deleting from $M_i$ all those edges that do not belong
to any filled triangle.  Clearly $E_i$ is still an independent-filled-hitting 
set, and since it contains exactly one edge of each filled triangle, its
size is also at most $f_G$.  The best of these three
disjoint independent-filled-hitting sets
contains at most $\frac{2n-5}{3}$ edges, due to the following:

\begin{lemma}
Any triangulated planar graph $G$ with $n \geq 4$
has at most $2n-5$ edges that belong to a filled triangle.
\end{lemma}

We would like to mention first that Cardinal et al.~\cite{CHK+15} 
gave a very similar result, namely that every triangulated planar graph contains
at most $2n-7$ edges that belong to a separating triangle.  (This immediately
implies that at most $2n-4$ edges belong to a filled triangle.)  Their proof
(not given in \cite{CHK+15}, but kindly shared via private communication) is
quite different from ours below, and does not seem to adapt easily to 
filled triangles.  Since our proof is quite short, we give it below despite
the rather minor improvement.

\begin{proof}
For every edge $e$ that belongs to a filled triangle,  fix an arbitrary
filled triangle $T$ containing $e$ and assign to $e$
the face that is incident to $e$ and inside triangle $T$.  We claim that
no face $f$ can have been assigned to two edges $e_1$ and $e_2$.  
Assume for contradiction that it did, so there are two distinct filled triangles
$T_1$ and $T_2$, both having $f$ inside and with $e_i$ incident to $T_i$
for $i=1,2$.  
Since face $f$ is inside both $T_1$ and $T_2$, one of the two triangles 
(say $T_2$) is inside the other (say $T_1$).  
Since $e_1$ belongs to $f$, it is on or inside $T_2$,
but since it is also on $T_1$ (which contains $T_2$ inside) it therefore
must be on $T_2$.  But then $T_2$ contains $e_1$ and $e_2$, and these two
(distinct) edges hence determine the three vertices of $T_2$.  But these
three vertices also belong to the triangular face $f$, and so $T_2$ is
an inner face and hence not a filled triangle by $n\geq 4$: contradiction.

With that, we can assign a unique inner face to every edge of a filled 
triangle, therefore in total there are at most $2n-5$ of them.
\end{proof}

We could find the smallest of edge sets $E_1,E_2,E_3$
by 4-coloring the graph (which can be done in $O(n^2)$ 
time \cite{RSST97}), but a better approach is the following:  Compute the 
dual graph $G^*$, and assign
weight 1 to an edge $e^*$ if the corresponding edge $e$ in $G$ belongs to a
filled triangle; else assign weight 0 to $e^*$.    Now find a minimum-weight
perfect matching $M$ in $G^*$; this can be done in $O((n\log n)^{1.5}
\sqrt{\alpha(n,n)})$ time \cite{GT91} since we have $m\in O(n)$ and maximum
weight 1.
Deleting from $M$ all edges of weight 0
then gives an independent-filled-hitting set $E_f$, and it has size at most
$\min\{f_G,\frac{2n-5}{3}\}$ since one of the three 
perfect matching of $G^*$ induced by a 4-coloring
would have at most this weight.  If we allow ourselves a slightly worse
matching, then we can find it in linear time:  Baker \cite{Baker94} gave
a linear-time PTAS for finding a maximum matching in a planar graph, and 
it can be easily adapted to a PTAS for minimum-weight perfect matching in a planar
graph.  We can hence conclude:

\begin{corollary}
Every planar graph $G$ has an independent-filled-hitting set of size at most
$\frac{2n-5}{3}$. It can be found in $O((n\log n)^{1.5}\alpha(n,n))$ time
or approximated arbitrarily close in $O(n)$ time.
\end{corollary}

It seems quite plausible that such an independent-filled-hitting set could
be computed in linear time, for example by modifying the 
perfect-matching-algorithm for 3-regular biconnected planar graphs \cite{BBD+01}
to take into account 0-1-edge-weights, with edges having weight 1 only if
they occur in a non-trivial 3-edge-cut.  This remains for future work.

\end{appendix}
\fi

\section{Proof of Lemma~\ref{lem:move_RI}}
\label{sec:move_RI}

We aim to prove the following:

\begin{quotation}
	Let $\Gamma$ be a planar RI-drawing. 
	Let $x$ be an interior vertex of degree 4 with neighbours 
	$u_1,u_2,u_3,u_4$ that form a 4-cycle.
	Assume that none of $x,u_1,u_2,u_3,u_4$ share a grid-line.
	Then we can move $x$ to a point on grid-lines of its neighbours
	and obtain a planar RI-drawing.
\end{quotation}

We assume that the naming is such that $u_1,u_2,u_3,u_4$ are the
neighbours of $x$ in counter-clockwise order.  We will also use the
notation $R(u,v)$ for the (open) axis-aligned rectangle 
whose diagonally opposite
corners are $u$ and $v$.

Consider the four quadrants relative to $x$, using the
open sets.  Each neighbour of $x$ shares no grid-line with $x$ and hence
belongs to some quadrant.
We claim that (after suitable
renaming) $u_i$ is in quadrant $i$ for $i=1,2,3,4$.  For if any quadrant
is empty, then either all four of $u_1,u_2,u_3,u_4$ are within two consecutive
quadrants (in case of which $x$ is outside cycle $u_1,u_2,u_3,u_4$, violating
planarity), or two consecutive vertices of $u_1,u_2,u_3,u_4$ are in diagonally
opposite quadrants (in case of which $x$ is inside their rectangle-of-influence,
violating the condition of an RI-drawing).  So each quadrant contains at
least one of the four vertices, implying that each contains exactly one of them.

Consider the five columns of $x,u_1,u_2,u_3,u_4$; for ease of description
we will denote these columns by $1,\dots,5$ (in order from left to right),
even though their actual $x$-coordinates may be different.
Likewise let $1,\dots,5$ be the five rows of $x,u_1,u_2,u_3,u_4$.
Since $x$ has a neighbour in each quadrant, it must be at $(3,3)$.
The open set $R\big((2,2),(4,4)\big)$ contains none of $u_1,u_2,u_3,u_4$,
so the cycle $C_u$ formed by $u_1,u_2,u_3,u_4$ goes around
$R\big((2,2),(4,4)\big)$.   The only vertex inside $C_u$ is $x$,
which implies that no vertex other than $u_3$ or $u_4$ can be at $(2,2)$ or $(4,2)$.
But not both $u_3$ and $u_4$ can be in row $2$, so we may
assume that point $x':=(2,2)$ contains no vertex (the other case
is similar).
Move $x$ to $x'$, which puts it
on grid-lines of two of its neighbours. 

\begin{figure}[ht]
\hspace*{\fill}
\includegraphics[width=80mm]{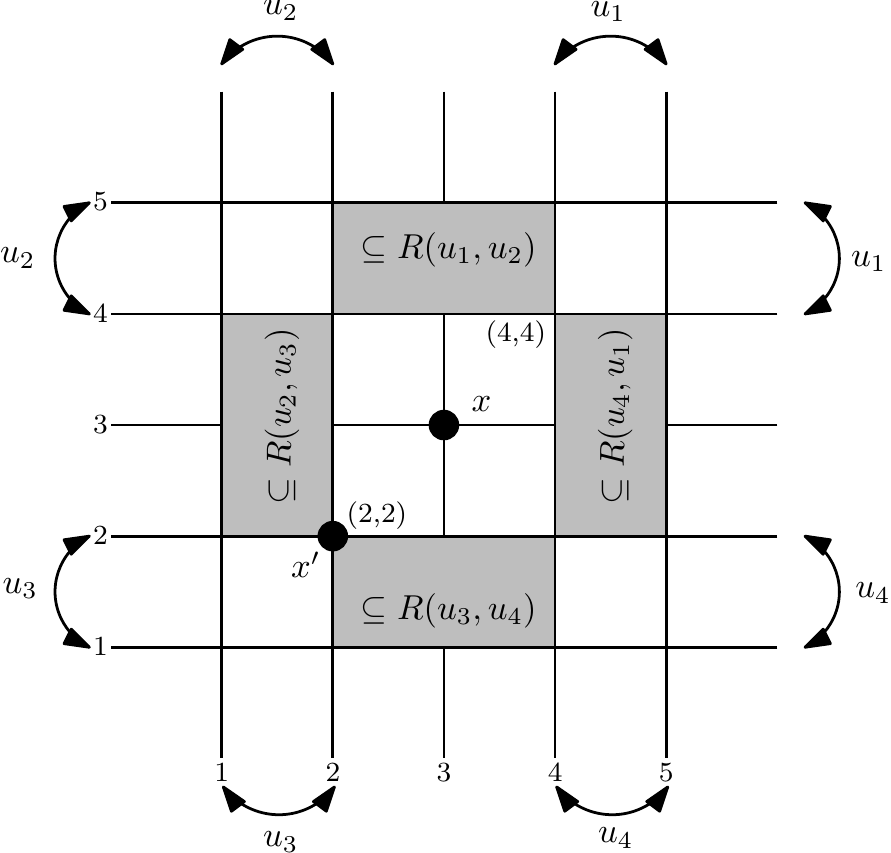}
\hspace*{\fill}
\caption{Moving $x$ to two grid-lines among its neighbours.}
\end{figure}

We claim that we obtain an RI-drawing,
and verify the conditions for the four edges $(x,u_i)$ separately:
\begin{itemize}
\item Vertex $u_3$ is on row 1 or 2 and
	column 1 or 2.  But by choice of $x'$ it is not at $(2,2)$.
	No matter where it is, rectangle $R(u_3,x)$ has $x'=(2,2)$
	inside or on the boundary,
	and therefore $R(u_3,x')\subset R(u_3,x)$ is empty since $(u_3,x)$
	is an edge in an RI-drawing.
\item Vertices $u_3$ and $u_4$ are on rows 1 and 2, but not on the
	same row, so $R(u_3,u_4)$ contains points that are between rows
	1 and 2.  Further, $u_3$ is on column 2 or left while $u_4$ is on
	column 4 or right, so $R(u_3,u_4)$ includes $R\big((2,1),(4,2)\big)$.
	So the empty rectangle $R(u_3,u_4)$ contains point $(2,2)=x'$ and
	therefore includes rectangle $R(x',u_4)$.  
\item Similarly one shows that $R(x',u_2)$ is
	empty.
\item It remains to show that $R(x',u_1)$ is empty, regardless of
	the position of $u_1$ within quadrant 1.  To do so, split
	$R(x',u_1)$ into parts and observe that all of them are empty.
	We already saw that $R_1:=R\big((2,2),(4,4)\big)$ is empty.
	We also claim that $R_2:=R\big((2,4),(4,5)\big)$ is empty.  This
	holds because $u_1$ and $u_2$ are on rows 4 and 5 (but not
	on the same row) and this rectangle hence is within $R(u_1,u_2)$.
	Similarly one shows that $R_3:=R\big((4,2),(5,4)\big)$ is empty.

	Notice that $R_1\cup R_2\cup R_3$ contains $R(x',u_1)$, unless
	$u_1$ is at $(5,5)$.  But in the latter case $R_4:=R\big((4,4),(5,5)\big)
	\subset R(x,u_1)$ is empty.  So either way $R(x',u_1)$ is contained
	within the union of empty rectangles and therefore is empty.
\end{itemize}

\end{document}